\documentclass[11pt]{article}

\usepackage[T1]{fontenc}
\usepackage[utf8]{inputenc}
\usepackage[a4paper,margin=1in]{geometry}
\usepackage{amsmath,amssymb,amsfonts,amsthm}
\usepackage{aliascnt}
\usepackage{authblk}
\usepackage{xcolor}
\usepackage{todonotes}
\usepackage{mathtools}
\usepackage{hyperref}
\usepackage[capitalize,noabbrev]{cleveref}

\hypersetup{
  pdftitle={A simple (2+ε)-approximation for knapsack interdiction},
  pdfauthor={Noah Weninger},
  pdfkeywords={interdiction, knapsack problem, approximation algorithms, bilevel optimization, dynamic programming}
}

\setuptodonotes{inline}

\theoremstyle{plain}
\newtheorem{theorem}{Theorem}
\newaliascnt{lemma}{theorem}
\newtheorem{lemma}[lemma]{Lemma}
\aliascntresetthe{lemma}
\newaliascnt{corollary}{theorem}
\newtheorem{corollary}[corollary]{Corollary}
\aliascntresetthe{corollary}
\newaliascnt{proposition}{theorem}
\newtheorem{proposition}[proposition]{Proposition}
\aliascntresetthe{proposition}

\numberwithin{equation}{section}

\def\OPTF{\operatorname{OPT}_{\mathrm{F}}}
\def\OPTIP{\operatorname{OPT}_{\mathrm{I}}}
\def\OPT{\operatorname{OPT}}
\def\F{F}
\def\FIP{K}
\def\Z{\mathbb{Z}}
\def\R{\mathbb{R}}
\def\se{\subseteq}

\newcommand{\PP}{{\mathrm{P}}}
\newcommand{\NP}{{\mathrm{NP}}}

\def\[#1\]{\begin{align*}#1\end{align*}}

\title{A simple \texorpdfstring{$(2+\epsilon)$}{(2+ε)}-approximation for knapsack interdiction}

\author{Noah Weninger\thanks{Current address: Department of Mathematics, Trier University, Germany.}}
\affil{Department of Combinatorics and Optimization, University of Waterloo, Canada\\
\href{mailto:nweninger@uwaterloo.ca}{nweninger@uwaterloo.ca}}
\date{\today}

\begin{document}

\maketitle

\begin{abstract}
In the knapsack interdiction problem, there are $n$ items, each with a non-negative profit, interdiction cost, and packing weight. There is also an interdiction budget and a capacity. The objective is to select a set of items to interdict (delete) subject to the budget which minimizes the maximum profit attainable by packing the remaining items subject to the capacity. We present a $(2+\epsilon)$-approximation running in $O(n^3\epsilon^{-1}\log(\epsilon^{-1}\log\sum_i p_i))$ time, where $p_i$ is the profit of item $i$. Although a polynomial-time approximation scheme (PTAS) is already known for this problem, our algorithm is considerably simpler and faster. The approach also generalizes naturally to a $(1+t+\epsilon)$-approximation for $t$-dimensional knapsack interdiction with running time $O(n^{t+2}\epsilon^{-1}\log(\epsilon^{-1}\log\sum_i p_i))$. In computational experiments on standard benchmark instances with $\epsilon=1$, our algorithm is on average within 0.4\% of optimal and over 20 times faster than a state-of-the-art exact algorithm.
\end{abstract}

\section{Introduction}

In an {\em interdiction problem}, the objective is to interdict (i.e., delete) a subset of resources from
a nominal combinatorial optimization problem, subject to an interdiction budget,
in order to maximally degrade the optimal objective value of the nominal problem.
Interdiction has been studied extensively for
knapsacks~\cite{denegre2011interdiction,weninger2024fast,chen2022approximation,caprara2014bilevel,caprara2016bilevel,dellacroce2020exact},
network flows~\cite{chestnut2017hardness,smith2020survey},
shortest paths~\cite{israeli2002shortest,khachiyan2008short},
matching~\cite{zenklusen2010matching},
matroids~\cite{linhares2017improved,weninger2025interdiction},
and cliques~\cite{furini2019maximum},
among many others.
These problems are natural models for defence, resource allocation, and competitive decision-making.

Knapsack interdiction is known to be $\Sigma_2^p$-hard~\cite{caprara2014bilevel}.
The same also holds for many other interdiction problems over $\NP$-hard nominal problems~\cite{grune2025complexity}.
Finding an exact solution in polynomial time is not possible unless the polynomial hierarchy collapses~\cite{caprara2014bilevel},
which motivates the study of approximation algorithms.
In this paper, we present a polynomial-time approximation algorithm for knapsack interdiction
which is significantly simpler and faster than previous approximations for this problem,
albeit with a slightly worse approximation ratio.

In recent years, interdiction problems have often been viewed as a special case of {\em bilevel optimization} (see e.g.~\cite{beck2026bilevel} for a general introduction).
In this context, the knapsack interdiction problem we consider is commonly known
in the literature as the ``bilevel knapsack problem with interdiction constraints''~\cite{weninger2024fast,caprara2016bilevel,dellacroce2020exact}.
Although there are many other problem variants which could perhaps be described by the name ``knapsack interdiction'',
we use this name for brevity, as it is the most well-studied variant.

The bilevel optimization model presents a definitional issue for the design of polynomial-time
approximation algorithms for $\Sigma_2^p$-hard interdiction problems.
A solution $(x,y)$ is {\em bilevel feasible} if $x$ is a feasible interdiction set and
$y$ is {\em optimal} for the nominal problem induced by $x$
(i.e., the nominal problem but with resources in $x$ deleted) \cite{beck2026bilevel}. Hence, for an $\NP$-hard nominal problem,
unless $\PP=\NP$, it is generally not possible to even produce a bilevel {\em feasible} solution in polynomial time,
let alone one with a bounded approximation ratio.
As such, our algorithm only outputs an interdiction set $x$;
every $y$ such that $(x,y)$ is bilevel feasible is guaranteed to achieve
the desired approximation ratio, but
we cannot find such a $y$ in polynomial time unless $\PP=\NP$.
This is the same convention taken in prior works on approximating
knapsack interdiction~\cite{chen2022approximation,caprara2014bilevel}.
If a bilevel feasible $(x,y)$ is needed,
the induced nominal problem can be solved separately after $x$ has been found;
\cref{lem:approx-framework} formalizes this.

After introducing the problem and related work,
in \cref{sec:pseudo} we apply linear programming (LP) duality and dynamic programming (DP) to obtain a pseudopolynomial-time 2-approximation.
\Cref{sec:fptas} modifies this algorithm into a polynomial-time $(2+\epsilon)$-approximation,
which naturally generalizes the standard fully polynomial-time approximation scheme (FPTAS) for 0-1 knapsack~\cite{lawler1979fast,kellerer2004knapsack}.
Although a polynomial-time approximation scheme (PTAS) is already known for knapsack interdiction~\cite{chen2022approximation},
our algorithm is considerably simpler and asymptotically faster:
it runs in time $O(n^3\epsilon^{-1}\log(\epsilon^{-1}\log\sum_i p_i))$ versus
$n^{\tilde O(\epsilon^{-4})}$ for the PTAS, making it well-suited for practical use.
\Cref{sec:multi} extends our result to a $(1+t+\epsilon)$-approximation for
$t$-dimensional knapsack interdiction.
Finally, \cref{sec:computational} evaluates the algorithm in computational experiments.

\subsection{Problem statement}\label{sec:problem}
An instance of the knapsack interdiction problem consists of $n$ items.
Each item $i\in[n]$ has profit $p_i\in\Z_{\ge0}$,
interdiction cost $c_i\in\Z_{\ge0}$, and weight $w_i\in\Z_{\ge0}$.
The interdiction budget is $B\in\Z_{\ge0}$ and the capacity is $C\in\Z_{\ge0}$.
The problem is to find
\[
\OPTIP=\min\{\FIP(x):x\in\{0,1\}^n,c^\top x\le B\}
\]
where $\FIP(x)=\max\{p^\top y:y\in\{0,1\}^n,w^\top y\le C,y\le\mathbf{1}-x\}$ and $\mathbf{1}$ is the $n$-dimensional all-ones vector.
We say that item $i$ is \emph{interdicted} in a solution $x$ if $x_i=1$;
in this case, the interdiction constraint $y_i\le1-x_i$ forces $y_i=0$.
We refer to the decision maker choosing $x$ as the \emph{interdictor}.

\subsection{Prior work}
Knapsack interdiction was introduced by DeNegre in 2011,
as a test case for a new branch-and-cut solver for bilevel mixed integer programming~\cite{denegre2011interdiction}.
A few years later, Caprara et al.~\cite{caprara2014bilevel} showed that knapsack interdiction
is $\Sigma_2^p$-complete, and is $\NP$-complete even in unary encoding, ruling out an FPTAS;
however, they gave a PTAS for the special subset-sum case, where $w_i=p_i$ for all $i$.
Chen et al.~\cite{chen2022approximation} later devised a PTAS for the general case.
Their result also extends to any constant number of knapsack constraints on $x$ or $y$,
achieving a $(t+\epsilon)$-approximation for any $\epsilon>0$ when there are $t$ knapsack constraints on $y$.
The running time is not explicitly bounded in their paper, but our analysis
determined it to be $n^{\tilde O(t\epsilon^{-2-2t})}$, making it prohibitively slow for practical use.

In concurrent and independent work, Sun and Li~\cite{sun2026approximation} claim a $(3+\epsilon)$-approximation running in $O(n^2)$ time, as well as a PTAS for bounded knapsack interdiction. However, their $(3+\epsilon)$-approximation appears to be incomplete as written: the output has no dependence on $\epsilon$.\footnote{
Their algorithm appears to solve the bilevel LP relaxation~\cite{carvalho2018polynomial}, then round down the single fractional component of $x$ to 0. However, let $M>1$ and consider the instance with $p=(M,1)$, $w=(1,1)$, $C=1$, $c=(M+1,1)$, and $B=M+1$. The output of their algorithm has value $M$, while the optimum is $1$, contradicting any constant-factor
approximation guarantee.}

Fractional knapsack interdiction, a key intermediate step in our algorithm,
is covered by the packing-LP interdiction framework
of Dinitz and Gupta~\cite{dinitz2013packing}, which gives a polynomial-time
$O(\log q\cdot\min\{q,\log k\})$-approximation for $(k,q)$-packing-LP interdiction, where
$k$ is the maximum number of variables in any one constraint and
$q$ is the maximum number of constraints containing any one variable of the packing LP.
The LP relaxation of $t$-dimensional knapsack interdiction that we use (see \cref{sec:pseudo,sec:multi})
is the $(n,t+1)$ case: each capacity constraint can involve all $n$ items, so $k=n$,
while each item lies in at most the $t$ capacity constraints together with its own upper bound $y_i\le1$, so $q=t+1$.
For any fixed $t$, their framework thus approximates this LP relaxation to within a constant factor,
whose dependence on $t$ scales as $(t+1)\log(t+1)$, and does not yield an FPTAS even when $t=1$.
In contrast, we give an FPTAS for any $t$.

Although exact knapsack interdiction solvers quickly solve many instances with up to 500 items~\cite{weninger2024fast,dellacroce2020exact},
some pathological instances (e.g., subset-sum) still cannot be solved reliably,
motivating the desire for practical approximation algorithms.
\Cref{sec:computational} shows empirically that our algorithm achieves this.
Our approximation algorithm is inspired by two techniques from recent works on knapsack interdiction:
guessing the split item~\cite{carvalho2018polynomial,dellacroce2020exact}
is analogous to guessing $\alpha$ in our algorithm,
and the use of DP is similar to the exact algorithm of Weninger and Fukasawa~\cite{weninger2024fast}.

Approximation algorithms have also been developed for many other interdiction problems, including
network flow~\cite{chestnut2017hardness},
shortest paths~\cite{khachiyan2008short},
matching~\cite{zenklusen2010matching,dinitz2013packing},
minimum spanning tree~\cite{frederickson1999increasing,zenklusen20151,linhares2017improved},
and connectivity~\cite{zenklusen2014connectivity,huang2025fptas}.
Inapproximability results are also known for many of these problems~\cite{caprara2014bilevel,khachiyan2008short,lee2017improved,linhares2017improved}.
General approximation frameworks for broader classes of structured interdiction problems have also been
proposed~\cite{chestnut2017interdicting,cong2026note}.
Of these, our approach shares perhaps the most similarity with Cong and Tian's framework for linear minimization interdiction~\cite{cong2026note},
in the sense that both approaches use the Lagrangian relaxation of a knapsack constraint to reduce the problem to 0-1 knapsack instances.
However, their problem is $\min$-$\min$ and they relax the interdiction budget constraint,
while ours is $\min$-$\max$ and we relax the capacity constraint.

\section{A pseudopolynomial-time 2-approximation}\label{sec:pseudo}

As a first step towards our main result, we present a pseudopolynomial-time 2-approximation.
Without loss of generality, we may assume $w_i\le C$ for all $i\in[n]$,
because all items violating this can be deleted without affecting the optimal solution.
With this assumption, we can obtain a 2-approximate upper bound $\OPTF$ by relaxing $\FIP(x)$ from $\{0,1\}$ to $[0,1]$:
\[\OPTF=\min\{\F(x):x\in\{0,1\}^n,c^\top x\le B\}\]
where
$\F(x)=\max\{p^\top y:y\in[0,1]^n,\,w^\top y\le C,\,y\le\mathbf{1}-x\}$.
For any fixed $x\in\{0,1\}^n$, $\F(x)$ is a fractional knapsack problem,
so $\FIP(x)\le\F(x)\le 2\FIP(x)$ since the LP relaxation of 0-1 knapsack
has an integrality gap of at most 2 under the assumption $w_i\le C$~\cite{kellerer2004knapsack}.
Consequently, $\OPTIP\le\OPTF\le2\OPTIP$ by taking the minimum over all $x$.
Although the worst-case gap between $\OPTF$ and $\OPTIP$ is a factor of 2,
\cref{prop:lp-gap} shows that this bound is much tighter for typical instances.

The main objective of the paper is to develop an FPTAS for $\OPTF$ (\cref{thm:kp-fptas}),
and hence a $(2+\epsilon)$-approximation for $\OPTIP$ (\cref{cor:2pluseps}).
We begin by showing that the 2-approximate relaxed problem $\OPTF$
can be solved in pseudopolynomial time.
Fix an interdiction $x$, so that $\F(x)$ is an LP.
We bound $\F(x)$ by taking the Lagrangian relaxation of the constraint $w^\top y\le C$ with multiplier $\alpha\ge0$.
For any feasible $y$ for $\F(x)$, the objective is bounded by
$p^\top y\le p^\top y+\alpha(C-w^\top y)=\alpha C+(p-\alpha w)^\top y$
since $C-w^\top y\ge0$.
Maximizing $\alpha C+(p-\alpha w)^\top y$ over $\{y\in[0,1]^n:y\le\mathbf{1}-x\}$,
which contains the feasible region of $\F(x)$, gives the upper bound
\[\F(x)\le\alpha C+\max\{(p-\alpha w)^\top y:y\in[0,1]^n,y\le\mathbf{1}-x\}.\]
In this maximum, each $y_i$ ranges independently over $[0,1-x_i]$,
so an optimal choice sets $y_i=1-x_i$ when $p_i-\alpha w_i>0$ and $y_i=0$ otherwise.
Therefore, the value of this maximum is \[\sum_{i\in[n]}(1-x_i)\max(0,p_i-\alpha w_i).\]
Writing $p'_i:=\max(0,p_i-\alpha w_i)$ (which depends on $\alpha$), and
\begin{equation}\label{eq:Falpha}
\F_\alpha(x):=\sum_{i\in[n]}(1-x_i)\max(0,p_i-\alpha w_i)=\sum_{i\in[n]}(1-x_i)p'_i,
\end{equation}
we obtain $\F(x)\le\alpha C+\F_\alpha(x)$ for every $\alpha\ge0$.
Instead of taking the Lagrangian relaxation of the constraint $w^\top y\le C$, one can take the LP dual
of $\F(x)$ and eliminate the dual variables of the interdiction constraints $y\le\mathbf{1}-x$;
the argument is virtually identical.

We claim that for any fixed $\alpha$, minimizing $\alpha C+\F_\alpha(x)$ over feasible interdictions $x$
reduces to a single 0-1 knapsack problem.
Indeed, $\F_\alpha(x)=\sum_{i\in[n]}p'_i-\sum_{i\in[n]}x_ip'_i$,
so, up to the additive constant $\sum_{i\in[n]}p'_i$,
minimizing $\F_\alpha(x)$ subject to $c^\top x\le B$ is the same as
maximizing $\sum_{i\in[n]}x_ip'_i$ subject to $c^\top x\le B$:
the 0-1 knapsack problem with item costs $c$, budget $B$, and profits $p'_i$,
which can be solved in pseudopolynomial time by DP~\cite{kellerer2004knapsack}.
Hence, writing
\begin{equation}\label{eq:g}
g(\alpha):=\alpha C+\min\{\F_\alpha(x):x\in\{0,1\}^n,c^\top x\le B\},
\end{equation}
the value $g(\alpha)$ is computable in pseudopolynomial time for any given $\alpha$.

Taking the minimum over all feasible $x$ on both sides of $\F(x)\le\alpha C+\F_\alpha(x)$
shows that $\OPTF\le g(\alpha)$ for every $\alpha\ge0$, so each $g(\alpha)$ is an upper bound on $\OPTF$.
Since the LP $\F(x)$ is feasible (e.g., take $y=0$) and bounded, the smallest
of the bounds $\alpha C+\F_\alpha(x)$ equals $\F(x)$ by strong duality:
\begin{equation}\label{eq:F-dual}
\min_{\alpha\ge0}\alpha C+\F_\alpha(x)=\F(x).
\end{equation}
By minimizing \eqref{eq:F-dual} over all feasible $x$ and exchanging the order of the two minimizations, we obtain
\[\OPTF=\min\{\alpha C+\F_\alpha(x):x\in\{0,1\}^n,c^\top x\le B,\alpha\ge0\}=\min_{\alpha\ge0}g(\alpha).\]
It suffices to evaluate $g$ at $O(n)$ candidate values:
by \cref{lem:lagrangian}, this minimum is attained at some $\alpha\in\{0\}\cup\{p_i/w_i:i\in[n]\}$
(with the convention $p_i/w_i=\infty$ when $w_i=0$).
\begin{lemma}
\label{lem:lagrangian}
For any $\alpha\ge0$ and $x\in\{0,1\}^n$, $\F(x)\le \alpha C+\F_\alpha(x)$.
For each $x$, equality is achieved by some $\alpha\in\{0\}\cup\{p_i/w_i:i\in[n]\}$.
\end{lemma}
\begin{proof}
Since $\F(x)=\min_{\alpha\ge0}\alpha C+\F_\alpha(x)$,
we have $\F(x)\le\alpha C+\F_\alpha(x)$ for any $\alpha\ge0$.
For each $i$, $\max(0,p_i-\alpha w_i)$ is linear in $\alpha$ except at $\alpha=p_i/w_i$,
so $\F_\alpha(x)$ is linear in $\alpha$ on every interval not containing $p_i/w_i$ for any $i$.
Hence, $\alpha C+\F_\alpha(x)$ is piecewise linear in $\alpha$ with breakpoints
in $\{p_i/w_i:i\in[n]\}$.
Therefore, the minimum is achieved at some $\alpha \in \{0\}\cup\{p_i/w_i:i\in[n]\}$.
\end{proof}
These $\alpha$ values correspond to the vertices of the feasible region of the LP dual of $\F(x)$.
This reduces computing $\OPTF$ to solving $O(n)$ integer programs $g(\alpha)$; in particular, $\OPTF$ is in $\NP$.
In contrast, $\OPTIP$ is not in $\NP$ unless the polynomial hierarchy collapses~\cite{caprara2014bilevel}.
We now obtain a pseudopolynomial-time 2-approximation for $\OPTIP$ via the following corollary.
\begin{corollary}
\label{cor:2approx}
There is a pseudopolynomial-time algorithm for $\OPTF$.
\end{corollary}
\begin{proof}
By the discussion above, $\OPTF=\min_{\alpha\ge0}g(\alpha)$ and, by \cref{lem:lagrangian},
this minimum is attained at one of the $O(n)$ values $\alpha\in\{0\}\cup\{p_i/w_i:i\in[n]\}$.
Since each $g(\alpha)$ is computable in pseudopolynomial time, evaluating $g$ at these values
and taking the minimum yields $\OPTF$.
\end{proof}
A similar algorithm can also be used for the continuous bilevel knapsack problem, in which both $x$ and $y$ are relaxed to $[0,1]^n$,
as demonstrated by Carvalho et al.~\cite{carvalho2018polynomial}.
In particular, for each choice of $\alpha$, the analogous $g(\alpha)$ (i.e., \eqref{eq:g} but where $x\in\{0,1\}^n$ is replaced
with $x\in[0,1]^n$) is then a fractional knapsack problem,
and therefore is computable in $O(n)$ time~\cite{kellerer2004knapsack}.
Enumerating all $O(n)$ values of $\alpha$ as in \cref{cor:2approx} then recovers the $O(n^2)$ algorithm of Carvalho et al.

$\OPTF$ also has the following additive approximation guarantee,
which is often much better than the multiplicative factor of 2.
\begin{proposition}\label{prop:lp-gap}
$\OPTIP\ge\OPTF-p^*$, where
\[p^*=\min_{x\in\{0,1\}^n}\left\{\max_{i\in[n]}(1-x_i)p_i:c^\top x\le B,\FIP(x)=\OPTIP\right\}.\]
\end{proposition}
\begin{proof}
Fix any $x\in\{0,1\}^n$ with $c^\top x\le B$. $\F(x)$ is a fractional knapsack problem, so it has an optimal solution $y$
with at most one fractional component~\cite{kellerer2004knapsack}.
Rounding the fractional component (if any) down to $0$ gives an integer $y'\le y$ which is feasible for $\FIP(x)$
with profit $p^\top y'\ge p^\top y-\max_i(1-x_i)p_i$. Since $y$ is optimal,
this implies $\FIP(x)\ge p^\top y'\ge\F(x)-\max_i(1-x_i)p_i$.
As $\F(x)\ge\OPTF$, we have $\FIP(x)\ge\OPTF-\max_i(1-x_i)p_i$ for every such $x$.
Applying it to each feasible $x$ with $\FIP(x)=\OPTIP$ and taking the one that minimizes $\max_i(1-x_i)p_i$ yields $\OPTIP\ge\OPTF-p^*$.
\end{proof}
For typical instances, $p^*$ is likely to be very small because
the formulation favours interdiction sets that remove high-profit items.
To achieve the worst-case 2-approximation, every optimal interdiction must leave behind
an item which is packed fractionally and has
profit comparable to $\OPTF$, which is unlikely to occur in practice,
as we demonstrate empirically in \cref{sec:computational}.
It is, however, possible to construct such an instance to demonstrate that our analysis is tight.
Let $M\ge3$ be an integer, $n=3$, $C=2M$ and $B=2$. Set $p_1=2M-1$, $w_1=2M$ and $c_1=2$, and for $i\in\{2,3\}$ set $p_i=w_i=M+1$ and $c_i=1$.
$\FIP$ is uniquely minimized by $x_I:=(1,0,0)$ with $\FIP(x_I)=M+1$, but $\F$ is uniquely minimized by $x_F:=(0,1,1)$
with $\F(x_F)=\FIP(x_F)=2M-1$, so the approximation factor $(2M-1)/(M+1)$ approaches $2$ as $M\to\infty$.

To obtain a bilevel feasible $\rho$-approximation, it suffices to minimize any $\rho$-approximation $\tilde\FIP$ of the nominal optimum $\FIP$ to obtain an interdiction set $\tilde x$, and then solve $\FIP(\tilde x)$.
For instance, $\tilde\FIP=\F$ satisfies $\FIP\le\F\le2\FIP$ and thus gives $\rho=2$.
\begin{lemma}\label{lem:approx-framework}
Let $\rho\ge1$ and suppose $\tilde\FIP:\{0,1\}^n\to\R_{\ge0}$ satisfies $\FIP(x)\le \tilde\FIP(x)\le \rho\FIP(x)$ for all $x\in\{0,1\}^n$ with $c^\top x\le B$.
Let $\tilde x\in\arg\min\{\tilde\FIP(x):x\in\{0,1\}^n,\,c^\top x\le B\}$,
and let $\tilde y\in\arg\max\{p^\top y:y\in\{0,1\}^n,w^\top y\le C,y\le\mathbf{1}-\tilde x\}$.
Then $(\tilde x,\tilde y)$ is bilevel feasible and $p^\top\tilde y\le\rho\OPTIP$.
\end{lemma}
\begin{proof}
By construction, $(\tilde x,\tilde y)$ is bilevel feasible and $p^\top\tilde y=\FIP(\tilde x)$.
Then, for every feasible $x$, $\tilde\FIP(\tilde x)\le \tilde\FIP(x)\le\rho\FIP(x)$.
Taking the minimum over $x$ gives $\tilde\FIP(\tilde x)\le\rho\OPTIP$
and hence $p^\top\tilde y=\FIP(\tilde x)\le\tilde\FIP(\tilde x)\le\rho\OPTIP$.
\end{proof}
Some variants and generalizations of this simple lemma have appeared previously~\cite{weninger2024fast,clemente2025local,wang2026randomized}.
Combining \cref{cor:2approx} with \cref{lem:approx-framework} (taking $\tilde\FIP=\F$, so $\rho=2$)
gives a bilevel feasible $2$-approximation for $\OPTIP$ in pseudopolynomial time.
In \cref{sec:fptas} we replace the pseudopolynomial DP with a polynomial-time approximation,
but the same caveat applies: this gives an approximate $x$ only, and to obtain a bilevel feasible $(x,y)$ from this,
the nominal problem must be solved to optimality afterwards, which cannot be done in polynomial time unless
$\PP=\NP$.

\section[A polynomial-time (2+epsilon)-approximation]{A polynomial-time $(2+\epsilon)$-approximation}\label{sec:fptas}
In this section, we show that this 2-approximation algorithm can be improved to polynomial time for
any fixed $\epsilon>0$ while losing only $1+\epsilon$ in the approximation factor.
Although \cref{cor:2approx} reduces finding a 2-approximation to solving just $O(n)$ knapsack problems,
the interdiction setting prevents the direct application of
approximation algorithms for 0-1 knapsack, such as the standard FPTAS~\cite{lawler1979fast,kellerer2004knapsack},
because the approximation factor does not carry over.

For instance, in this standard knapsack FPTAS,
each item profit is rounded to a multiple of $\epsilon p_{\max}/n$
where $p_{\max}=\max_i p_i$; the total error is then at most $\epsilon p_{\max}\le\epsilon\OPT$
because we assume each item can be packed individually.
In knapsack interdiction, the interdictor can remove the most profitable items,
so $\OPTF$ may be much smaller than $p_{\max}$
and rounding profits this way gives inadequate precision.
That is, we need the approximation error to be bounded relative to the optimal solution
to the interdiction problem.
To achieve this, we guess some $z\approx\OPTF$,
round the profits to multiples of $\epsilon z/n$ and solve by DP for each $\alpha$;
the correct value of $z$ is found by binary search.

Two properties of this rounding scheme enable binary search to work here.
If $z$ is too small, the optimal solution is pruned from the DP,
but we can detect this situation.
This pruning also keeps the running time polynomial.
If $z$ is too large, the rounding is too coarse and the desired approximation ratio is not attained.
Binary search then determines the correct $z$.

\begin{theorem}\label{thm:kp-fptas}
For every $\epsilon\in(0,1]$, $\OPTF$ admits a $(1+\epsilon)$-approximation with running time
$O(n^3\,\epsilon^{-1}\log(\epsilon^{-1}\log\sum_i p_i))$ (i.e., an FPTAS).
\end{theorem}
\begin{proof}
We first describe how to achieve the desired approximation ratio,
and then conclude by showing that it can be done within the desired running time.

If the budget is large enough that every item with positive profit can be interdicted,
then $\OPTF=0$ and the problem is trivial. We therefore assume that this is not the case,
so that every feasible $x$ leaves some item $i$ with $p_i>0$ uninterdicted.
Then, packing such an item $i$ alone is feasible for the nominal problem
because $w_i\le C$. So, $\F(x)\ge p_i\ge1$ because $p_i$ is a positive integer,
and hence $\OPTF\ge1$.

The core of the algorithm rounds data relative to a guessed parameter $z\approx\OPTF$.
We will see how the correct $z$ is determined later; for now just assume $z>0$ and let $\delta:=\epsilon z/n$.
Recall from \eqref{eq:Falpha} that $\F_\alpha(x)=\sum_{i\in[n]}(1-x_i)p'_i$, where $p'_i=\max(0,p_i-\alpha w_i)$.
Let $\tilde p_i:=\lceil p'_i/\delta\rceil\delta$ be $p'_i$ rounded up to the nearest multiple of $\delta$
(note that $\tilde p_i$ depends on both $\alpha$ and $z$). Define
\begin{align*}
\tilde\F_\alpha(x):=\sum_{i\in[n]}(1-x_i)\tilde p_i,\text{ and }
~\tilde g(\alpha):=\alpha C+\min\{\tilde\F_\alpha(x):x\in\{0,1\}^n,c^\top x\le B\}
\end{align*}
as the rounded versions of $\F_\alpha(x)$ and $g(\alpha)$, respectively.
Note that each term of $\tilde\F_\alpha(x)$ is rounded up by less than $\delta$, there are at most $n$ terms,
and $n\delta=\epsilon z$, so $\F_\alpha(x)\le\tilde\F_\alpha(x)\le\F_\alpha(x)+\epsilon z$.
Taking the minimum over all feasible $x$ gives $g(\alpha)\le\tilde g(\alpha)\le g(\alpha)+\epsilon z$.

Since $\OPTF$ is not known in advance,
we search over $z=(1+\epsilon)^j$ for $j\in\{0,1,\dots,J\}$
where $J=\lceil\log_{1+\epsilon}\sum_ip_i\rceil$.
For a given candidate $z$, let
$\alpha^*\in\arg\min\{\tilde g(\alpha):\alpha\in\{0\}\cup\{p_i/w_i:i\in[n]\}\}$;
we say that $z$ is {\em accepted} if $\tilde g(\alpha^*)\le(1+\epsilon)z$.
We claim that every $z\ge\OPTF$ is accepted,
and every $z<\OPTF/(1+\epsilon)$ is rejected.
Let $\hat x$ achieve $\OPTF$ and let $\hat\alpha$ achieve equality in
\cref{lem:lagrangian}, so that $\F(\hat x)=g(\hat\alpha)=\OPTF$. When $z\ge\OPTF$, the test passes, since
\[\tilde g(\alpha^*)\le\tilde g(\hat\alpha)\le g(\hat\alpha)+\epsilon z=\OPTF+\epsilon z\le(1+\epsilon)z.\]
When $z<\OPTF/(1+\epsilon)$, the test fails, since
\[\tilde g(\alpha^*)\ge g(\alpha^*)\ge\OPTF>(1+\epsilon)z.\]
This proves the claim.

Now, since consecutive candidates differ by a factor $1+\epsilon$,
at most one candidate lies in the remaining interval $[\OPTF/(1+\epsilon),\OPTF)$,
so there is some threshold $T$ such that $z\ge T$ is accepted and $z<T$ is rejected.
Binary search therefore finds the smallest accepted candidate $z^*$ in $O(\log J)$ steps,
and $z^*\le(1+\epsilon)\OPTF$, since at least one accepted candidate lies in $[\OPTF,(1+\epsilon)\OPTF)$.

Let $\tilde x$ attain $\min\{\tilde\F_{\alpha^*}(x):x\in\{0,1\}^n,c^\top x\le B\}$ at $z=z^*$.
Since $\F(x)\le\alpha C+\F_\alpha(x)$ for any $\alpha$ by \cref{lem:lagrangian},
\begin{align*}
\F(\tilde x) &\le\alpha^*\,C+\F_{\alpha^*}(\tilde x)\\
&\le\alpha^*\,C+\tilde\F_{\alpha^*}(\tilde x)\\
&=\tilde g(\alpha^*)\\
&\le(1+\epsilon)\,z^*\\
&\le(1+\epsilon)^2\,\OPTF.
\end{align*}
Setting $\epsilon'=\sqrt{1+\epsilon}-1$ and using $\epsilon'$ in place of $\epsilon$
gives $\F(\tilde x)\le(1+\epsilon')^2\OPTF=(1+\epsilon)\OPTF$, a $(1+\epsilon)$-approximation for $\OPTF$.

It remains to bound the running time.
The standard pseudopolynomial DP for knapsack discussed in \cref{cor:2approx}
has $O(nB)$ states, identified by a pair $(i,B')$ describing the item index and remaining budget, respectively.
To achieve polynomial time, we instead use a DP with states identified by the item index and accumulated rounded profit,
similar to the standard knapsack FPTAS~\cite{lawler1979fast}.
Specifically, $f_\alpha(i,\tilde P)$ stores the minimum interdiction cost needed to ensure
$\tilde\F_\alpha$ on items $i,\dots,n$ is at most $\tilde P$:
\[
f_\alpha(i,\tilde P)=
\begin{dcases}
0&\text{if }i>n\text{ and }\tilde P\ge0,\\
\infty&\text{if }i>n\text{ and }\tilde P<0,\\
\min\left\{\begin{aligned}
& c_i+f_\alpha(i+1,\tilde P),\\
& f_\alpha(i+1,\tilde P-\tilde p_i)
\end{aligned}\right\}&\text{if }i\le n.
\end{dcases}
\]
The first branch of the minimum interdicts item $i$, incurring the interdiction cost $c_i$; the second does not,
and reduces the target $\tilde P$ by the rounded profit $\tilde p_i$.
Since each $\tilde p_i$ is a multiple of $\delta$, the invariant that $\tilde P$ is a multiple of $\delta$ is maintained.
Therefore,
\begin{equation}\label{eq:gtilde}
\tilde g(\alpha)=\alpha C+\min\{\tilde P\in\{0,\delta,2\delta,\dots\}:f_\alpha(1,\tilde P)\le B\},
\end{equation}
where the minimum is the smallest rounded profit target $\tilde P$ attainable within the budget.
We claim that it suffices to consider only a polynomial number of such values of $\tilde P$, namely those with $\tilde P\le(1+\epsilon)z-\alpha C$.
To this end, define
\[\hat g(\alpha):=\alpha C+\min\{\tilde P\in\{0,\delta,2\delta,\dots\}:f_\alpha(1,\tilde P)\le B,\,\tilde P\le(1+\epsilon)z-\alpha C\},\]
with the convention $\min\emptyset=\infty$.
When $\tilde g(\alpha)\le(1+\epsilon)z$, we have $\hat g(\alpha)=\tilde g(\alpha)$ because the $\tilde P$
minimizing \eqref{eq:gtilde} satisfies $\tilde P=\tilde g(\alpha)-\alpha C\le(1+\epsilon)z-\alpha C$.
If instead $\tilde g(\alpha)>(1+\epsilon)z$, then $\hat g(\alpha)=\infty$
because any feasible $\tilde P$ in \eqref{eq:gtilde} satisfies $\tilde P\ge\tilde g(\alpha)-\alpha C>(1+\epsilon)z-\alpha C$,
which violates the constraint $\tilde P\le(1+\epsilon)z-\alpha C$ in $\hat g$.
Recall that $z$ is accepted if $\tilde g(\alpha^*)\le(1+\epsilon)z$.
So, when $z$ is accepted, $\hat g(\alpha^*)=\tilde g(\alpha^*)$ and $\tilde g\le\hat g$ everywhere,
so they have the same minimizing $\alpha^*$.
If $z$ is rejected, then $\tilde g(\alpha)\ge\tilde g(\alpha^*)>(1+\epsilon)z$ for all $\alpha$,
so we also have $\hat g(\alpha)=\infty>(1+\epsilon)z$.
Therefore, it suffices to compute $\hat g$ in place of $\tilde g$.

We can now compute the final running time bound.
Recall that we run the algorithm with $\epsilon'=\sqrt{1+\epsilon}-1$ in place of $\epsilon$.
Since $\epsilon\in(0,1]$, we have $\epsilon/(1+\sqrt2)\le\epsilon'\le\epsilon$ and so
bounding the running time in terms of $\epsilon$ only loses a constant factor.
Since $\alpha C\ge0$, restricting the DP to states with $\tilde P\le (1+\epsilon)z-\alpha C$ ensures there are
at most $(1+\epsilon)z/\delta=n\epsilon^{-1}+n$ such values of $\tilde P$ considered, which is $O(n\epsilon^{-1})$
since $\epsilon\in(0,1]$. So, there are at most $O(n^2\epsilon^{-1})$ total states in the DP,
and for each $z$ we solve this DP at $n+1$ values of $\alpha$, costing $O(n^3\epsilon^{-1})$ in total.
Finally, the binary search for $z^*$ needs $O(\log J)$ steps.
Since $\epsilon\in(0,1]$, we have
$J\le\log_{1+\epsilon}(\sum_i p_i)+1=O(\epsilon^{-1}\log(\sum_i p_i))$,
so the binary search takes $O(\log(\epsilon^{-1}\log\sum_i p_i))$ steps.
The overall running time is therefore bounded by
$O(n^3\,\epsilon^{-1}\log(\epsilon^{-1}\log\sum_ip_i))$.
\end{proof}
\begin{corollary}\label{cor:2pluseps}
For every $\epsilon\in(0,1]$, $\OPTIP$ admits a $(2+\epsilon)$-approximation with running time
$O(n^3\,\epsilon^{-1}\log(\epsilon^{-1}\log\sum_i p_i))$.
\end{corollary}
\begin{proof}
Run the FPTAS of \cref{thm:kp-fptas} with parameter $\epsilon/2$ to obtain $\tilde x$
with $\F(\tilde x)\le(1+\epsilon/2)\OPTF$. Then
\[\FIP(\tilde x)\le\F(\tilde x)\le(1+\epsilon/2)\OPTF\le(1+\epsilon/2)\cdot2\OPTIP=(2+\epsilon)\OPTIP,\]
where the first inequality is $\FIP\le\F$, the second is \cref{thm:kp-fptas} applied with $\epsilon/2$,
and the third is $\OPTF\le2\OPTIP$. Thus $\tilde x$ is a $(2+\epsilon)$-approximation:
$\FIP(\tilde x)\le(2+\epsilon)\OPTIP$.
\end{proof}
We remark that our results require $\epsilon\le1$ only to simplify the presentation.
Larger $\epsilon$ values can also be used, but the running time bound then no longer decreases
proportional to $\epsilon^{-1}$.
For instance, the algorithm of \cref{cor:2pluseps} still attains a $(2+\epsilon)$-approximation for $\epsilon>1$,
and the running time remains bounded by $O(n^3\log(\log\sum_i p_i))$ with a constant independent of $\epsilon$.

\section{Multidimensional knapsack interdiction}\label{sec:multi}

The approach extends to $t$ capacity constraints (i.e., where the nominal problem is $t$-dimensional knapsack).
In this setting, each item $i$ has a weight vector $w_i\in\Z_{\ge0}^t$
and the capacity is a vector $C\in\Z_{\ge0}^t$;
the interdiction costs $c_i$ and budget $B$ remain scalar.
Let $W\in\Z_{\ge0}^{t\times n}$ be the matrix whose $i$-th column is $w_i$,
and assume $w_{ij}\le C_j$ for all $i\in[n]$ and $j\in[t]$.
The definitions of $\FIP(x)$ and $\F(x)$ generalize to
\begin{align*}
\FIP(x)&=\max\{p^\top y:Wy\le C,\,y\le\mathbf{1}-x,\,y\in\{0,1\}^n\},\\
\F(x)&=\max\{p^\top y:Wy\le C,\,y\le\mathbf{1}-x,\,y\in[0,1]^n\}.
\end{align*}

\begin{corollary}\label{cor:multi}
For every $\epsilon\in(0,1]$, a $(1+t+\epsilon)$-approximation for $\OPTIP$
with $t$ capacity constraints
can be computed in $O(n^{t+2}\,\epsilon^{-1}\log(\epsilon^{-1}\log\sum_i p_i))$ time.
\end{corollary}
\begin{proof}
The LP relaxation of $t$-dimensional 0-1 knapsack has integrality gap at most $1+t$
under the assumption $w_{ij}\le C_j$~\cite{kellerer2004knapsack},
so $\F(x)\le(1+t)\,\FIP(x)$ and hence $\OPTF\le(1+t)\,\OPTIP$.

Next, we $(1+\epsilon)$-approximate $\OPTF$.
Taking the Lagrangian relaxation of the $t$ capacity constraints
analogously to the single-constraint case gives
\[\OPTF=\min\{\alpha^\top C+\F_\alpha(x):
  c^\top x\le B,\,x\in\{0,1\}^n,\,\alpha\in\R_{\ge0}^t\},\]
where $\F_\alpha(x)=\sum_{i\in[n]}(1-x_i)\max(0,p_i-\alpha^\top w_i)$
and $\alpha\in\R_{\ge0}^t$ is a vector of dual variables, one per capacity constraint.

As in \cref{lem:lagrangian}, we now show that $\alpha$ can be restricted to a finite set which does not depend on $x$.
For each $S\se[n]$, let
\[P_S:=\{\alpha\in\R^t_{\ge0}:\alpha^\top w_i\le p_i\text{ for }i\in S,\text{ and }\alpha^\top w_i\ge p_i\text{ for }i\notin S\}.\]
For every $\alpha\in P_S$, we have $p_i-\alpha^\top w_i\ge0$ for $i\in S$ and $p_i-\alpha^\top w_i\le0$ for $i\notin S$,
so $\F_\alpha(x)=\sum_{i\in S}(1-x_i)(p_i-\alpha^\top w_i)$.
Therefore, $\min\{\alpha^\top C+\F_\alpha(x):\alpha\in P_S\}$ is an LP for each $S$.
If $P_S$ is nonempty, this LP attains an optimal solution at a vertex of $P_S$
because its objective is nonnegative and $P_S\se\R^t_{\ge0}$ contains no line.

Note that the set of equations obtained by making a constraint of $P_S$ tight, that is,
\[\{p_i=\alpha^\top w_i:i\in[n]\}\cup\{\alpha_j=0:j\in[t]\},\]
does not depend on $S$, because $S$ only changes the direction of the inequalities.
Every vertex of $P_S$ is the unique solution of some $t$ of these $n+t$ equations.
Call any point of $\R^t_{\ge0}$ with this property a {\em candidate}, and let $A$ be the set of candidates, so that every vertex of every $P_S$ lies in $A$.
Since each of the $\binom{n+t}{t}$ choices of $t$ equations determines at most one candidate, $|A|\le\binom{n+t}{t}=O(n^t)$.
Combining these facts, we obtain
\begin{align}
\OPTF&=\min\{\alpha^\top C+\F_\alpha(x):c^\top x\le B,x\in\{0,1\}^n,\alpha\in\R^t_{\ge0}\}\nonumber\\
&=\min\{\alpha^\top C+\F_\alpha(x):c^\top x\le B,x\in\{0,1\}^n,S\se[n],\alpha\in P_S\}\label{eq:cover}\\
&=\min\{\alpha^\top C+\F_\alpha(x):c^\top x\le B,x\in\{0,1\}^n,\alpha\in A\},\label{eq:candidates}
\end{align}
where \eqref{eq:cover} holds because every $\alpha\in\R^t_{\ge0}$ lies in $P_S$ for $S=\{i\in[n]:\alpha^\top w_i\le p_i\}$,
and \eqref{eq:candidates} holds because when $\min\{\alpha^\top C+\F_\alpha(x):\alpha\in P_S\}$ has an optimal solution,
it has one at some $\alpha\in A$.
The set $A$ can be computed by solving each system of $t$ of the $n+t$ equations.
For example, when $t=1$, each nonempty $P_S$ is an interval with endpoints in $A=\{0\}\cup\{p_i/w_i:i\in[n],w_i>0\}$, just as in \cref{lem:lagrangian}.

For each $\alpha\in A$, $\min\{\F_\alpha(x):x\in\{0,1\}^n,\,c^\top x\le B\}$
reduces to 0-1 knapsack as in \cref{cor:2approx}, so the profit-rounding FPTAS of \cref{thm:kp-fptas} applies directly,
with $O(n^t)$ candidate values of $\alpha$ replacing the $n+1$ in the single-constraint case.
This yields a $(1+\epsilon)$-approximation for $\OPTF$
in $O(n^{t+2}\,\epsilon^{-1}\log(\epsilon^{-1}\log\sum_i p_i))$ time.
Running this algorithm with $\epsilon/(1+t)$ in place of $\epsilon$ yields
\[\F(\tilde x)\le(1+\epsilon/(1+t))\OPTF\le(1+t)(1+\epsilon/(1+t))\OPTIP=(1+t+\epsilon)\OPTIP.\tag*{\qedhere}\]
\end{proof}
As in the single-constraint case, the algorithm of Chen et al.~\cite{chen2022approximation}
achieves a superior approximation ratio of $t+\epsilon$ instead of our $1+t+\epsilon$,
but its $n^{\tilde O(t\epsilon^{-2-2t})}$ running time is impractical even for modest $t$ and $\epsilon$.
On the other hand, the running time of our algorithm is polynomial in $n$ and $1/\epsilon$ for every fixed $t$.

\section{Computational experiments}\label{sec:computational}

To demonstrate the typical approximation quality and running time of our approach,
we performed computational experiments on a standard set of problem instances from the literature.
The algorithm of \cref{cor:2pluseps} was implemented in
Rust and consists of about 250 lines of code in
total.\footnote{The source code is available at \url{https://github.com/nwoeanhinnogaehr/ki-approx}.}
We compare this implementation against the
exact solver of Weninger and Fukasawa
on the core benchmark instances used in their paper~\cite{weninger2024fast},
which were aggregated from the works of DeNegre~\cite{denegre2011interdiction},
Caprara et al.~\cite{caprara2016bilevel},
Tang et al.~\cite{tang2016class},
Fischetti et al.~\cite{fischetti2019interdiction},
and Della Croce and Scatamacchia~\cite{dellacroce2020exact}.
These instances have between 10 and 500 items each.
We do not compare computationally against the PTAS of Chen et al.~\cite{chen2022approximation},
since implementing it would be highly involved,
and we expect its running time to be prohibitive on instances of this size
when its parameter is chosen small enough to guarantee an approximation ratio comparable to ours.

All experiments were run on a machine with four 16-core Intel Xeon Gold 6142 CPUs at 2.60 GHz and 256~GB of RAM.
The approximation algorithm implementation is single-threaded.
On the other hand, the exact algorithm was run with 16 threads to match the original paper~\cite{weninger2024fast},
but only the initial DP phase is parallelized;
the branch-and-bound phase, which dominates the running time on the hardest instances, is serial.
We remark that the approximation algorithm could also be parallelized easily because the DPs for different values of $\alpha$ are independent,
but we use a single thread so as not to overstate its advantage.

To evaluate the solution quality after finding $\tilde x$, we obtain a bilevel feasible solution as in \cref{lem:approx-framework}
by solving the nominal problem $\FIP(\tilde x)$ using the standard knapsack DP~\cite{kellerer2004knapsack}.
We also solve the relaxed nominal problem $\F(\tilde x)$ using the standard greedy algorithm~\cite{kellerer2004knapsack}.
The running times reported exclude the time spent to compute $\FIP(\tilde x)$ and $\F(\tilde x)$; however, this time
was at most 1.1 ms in the worst case.

\Cref{tab:eps} summarizes the results for $\epsilon\in\{0.01,0.1,1,10,100\}$.
This table excludes 16 of the Fischetti et~al.~instances which the exact algorithm did not solve within one hour,
and 50 trivial instances with $\OPTIP=0$. This leaves 1274 instances in total.
For each $\epsilon$, the ``optimal'' column gives the percentage of instances on which $\tilde x$ is optimal,
that is, $\FIP(\tilde x)=\OPTIP$.
The columns $\FIP(\tilde x)/\OPTIP$ and $\F(\tilde x)/\OPTIP$ give the
approximation ratio of the bilevel feasible solution and relaxed nominal solution obtained from $\tilde x$, respectively.
The ``time'' columns indicate the running time of the approximation algorithm,
and the ``speedup'' columns report the running time of the exact algorithm divided by that of the approximation algorithm.
All averages are reported as geometric means.

\begin{table}[t]
\centering
\caption{Computational results on all 1274 instances for various choices of $\epsilon$.}
\label{tab:eps}
\medskip
\begin{tabular}{r|r|rr|rr|rr|rr}
& & \multicolumn{2}{c|}{$\FIP(\tilde x)/\OPTIP$} & \multicolumn{2}{c|}{$\F(\tilde x)/\OPTIP$} & \multicolumn{2}{c|}{time} & \multicolumn{2}{c}{speedup} \\
$\epsilon$ & optimal & mean & max & mean & max & mean & max & mean & min \\
\hline
100 & 23.5\% & 1.0494 & 2.0772 & 1.0612 & 2.0776 & 1.45 ms & 25.5 ms & 37.9 & 2.28 \\
10 & 35.4\% & 1.0119 & 1.5000 & 1.0222 & 1.6942 & 1.51 ms & 27.6 ms & 36.4 & 2.80 \\
1 & 67.7\% & 1.0039 & 1.5000 & 1.0145 & 1.6942 & 2.51 ms & 73.6 ms & 21.9 & 1.34 \\
0.1 & 88.9\% & 1.0037 & 1.5000 & 1.0141 & 1.6942 & 8.40 ms & 856 ms & 6.54 & 0.226 \\
0.01 & 89.7\% & 1.0037 & 1.5000 & 1.0141 & 1.6942 & 56.7 ms & 10.9 s & 0.969 & 0.0187 \\
\end{tabular}
\end{table}

Overall, the solutions achieve a much better approximation ratio than the worst case of $2+\epsilon$.
Even at $\epsilon=1$, $\FIP(\tilde x)$ is within 0.4\% of $\OPTIP$ on average,
and $\tilde x$ is optimal for 67.7\% of the instances.
The worst approximation ratio at $\epsilon\le10$ is 1.5, which is attained on a single instance and does not improve as
$\epsilon$ decreases, so it is caused by the gap between $\OPTF$ and $\OPTIP$ rather than by rounding.
Excluding the 10 worst instances, the approximation ratio is at most 1.2 for $\epsilon\le10$.
This supports our claim following \cref{prop:lp-gap} that the worst case is unlikely to occur in practice.

For $\epsilon\ge0.1$, our approximation algorithm is faster on average and more predictable than the exact algorithm.
At $\epsilon=1$, the approximation algorithm is 21.9 times faster on average, and at least 1.3 times faster
on every instance.
Even at $\epsilon=0.01$, the approximation algorithm is just over 3\% slower on average.
As noted in the discussion following \cref{cor:2pluseps}, the running time is not expected to improve proportionally to
$\epsilon^{-1}$ as $\epsilon$ is increased above 1. This is reflected in the data: the average running time is 2.51 ms at $\epsilon=1$, but only improves to 1.51 ms at $\epsilon=10$ and 1.45 ms at $\epsilon=100$.
Since the observed solution quality does continue to worsen significantly as $\epsilon$ increases, large values
of $\epsilon$ are likely poor choices in practice.

Although the exact algorithm is typically fast on these benchmarks, taking 54.9 ms on average,
it struggles on the 100 subset-sum instances of Fischetti et al.~\cite{fischetti2019interdiction}:
two of them take several minutes and 16 of them time out after one hour.
The approximation algorithm offers a clear advantage on these instances.
At $\epsilon=1$, it takes at most 7 ms and solves all subset-sum instances to within 0.3\% of optimality (for unsolved instances,
we measure against the lower bound computed by the exact algorithm).
At $\epsilon=0.1$, it takes at most 31 ms, solves to optimality all 84 instances that the exact algorithm solved,
and remains within 0.3\% of optimal on the 16 unsolved ones.

Overall, for $\epsilon$ between $0.1$ and $1$, our algorithm is both faster than the state-of-the-art exact algorithm and
nearly optimal on standard benchmarks, making it a practical alternative in applications where strict optimality is not required.

\section{Conclusion}\label{sec:approx-conclusion}

We gave a polynomial-time $(2+\epsilon)$-approximation for knapsack interdiction,
and more generally a $(1+t+\epsilon)$-approximation for $t$-dimensional knapsack interdiction.
Both results rely on approximating the bilevel optimum $\OPTIP$ by the LP-relaxed value $\OPTF$,
which is within a factor $1+t$ of $\OPTIP$ and, unlike $\OPTIP$, lies in $\NP$.

Our algorithm attains a weaker approximation ratio than the PTAS of
Chen et al.~\cite{chen2022approximation} ($1+t+\epsilon$ rather than $t+\epsilon$), but is far faster:
its running time $O(n^{t+2}\epsilon^{-1}\log(\epsilon^{-1}\log\sum_ip_i))$ is polynomial in $n$
for every fixed $t$, with $\epsilon$ appearing outside the exponent,
whereas their $n^{\tilde O(t\epsilon^{-2-2t})}$ running time is impractical even for modest $t$ and $\epsilon$.
Empirically, for $\epsilon$ between $0.1$ and $1$, our algorithm is nearly optimal on standard benchmark instances
and typically much faster than a state-of-the-art exact algorithm.
For future work, it would be interesting to try to push our approach further to obtain
an alternative PTAS with improved running time, or to determine whether any constant-factor
approximation can be obtained in linear or quadratic time. It would also be interesting to determine if
the formulation of the nominal problem can be strengthened to improve the approximation ratio below 2.
It may also be possible to generalize our
approach to handle multiple knapsack constraints on $x$, or other knapsack problem variants.

\paragraph*{Acknowledgements.}
I thank Ricardo Fukasawa, Chaitanya Swamy, and three anonymous referees for their helpful suggestions which improved the presentation and content of the paper.

\paragraph*{Funding.}
I acknowledge the support of the Natural Sciences and Engineering Research Council of Canada (NSERC), [funding reference number CGSD-2023-578589].

\bibliographystyle{plainurl}
\bibliography{references-sosa}

\end{document}